\documentclass[11pt]{article}
\usepackage{epsf}
\usepackage{epsfig}
\usepackage{graphicx}
\usepackage{color}
\usepackage{amsmath}
\usepackage{mathrsfs}

\usepackage{amsthm}
\usepackage{latexsym}
\usepackage{amssymb}
\usepackage{amsmath}

\usepackage{stmaryrd}

\usepackage{epsf}
\usepackage{verbatim}

\usepackage[margin=1in]{geometry}
\usepackage{xspace}
\usepackage{xcolor}
\usepackage{setspace}

\newtheorem{coro}{Corollary}[section]

\newtheorem{theorem}{Theorem}[section]
\newtheorem{definition}{Definition}[section]
\newtheorem{lemma}{Lemma}[section]

\catcode`@=11 \@addtoreset{equation}{section} 

\begin{document}

\title{{\bf On Blockwise Symmetric Matchgate Signatures and Higher Domain \#CSP}}

\vspace{0.3in}
\author{
Zhiguo Fu\thanks{School of Mathematics, Jilin University; {\tt fuzg@jlu.edu.cn}.}}

\date{}
\maketitle
\vspace{0.3in}
\begin{abstract}
For any $n\geq 3$ and $ q\geq 3$,
we prove that
the {\sc Equality}
function
$(=_n)$ on $n$ variables over a domain of size $q$
 cannot be realized by matchgates under holographic transformations.
This is a consequence of our theorem on the structure of blockwise symmetric matchgate signatures.
This has the implication that the standard holographic algorithms based on matchgates,
 a methodology known
to be universal for \#CSP over the Boolean domain,
cannot produce P-time algorithms for planar \#CSP over any
higher domain $q\geq 3$.
\end{abstract}

\thispagestyle{empty}
\newpage
\setcounter{page}{1}

\section{Introduction}
Half a century ago, the
Fisher-Kasteleyn-Temperley (FKT) algorithm
was discovered~\cite{TF61, Kasteleyn1961, Kasteleyn1967}. The FKT algorithm can
count the number of perfect matchings (dimers) over planar graphs
in polynomial time. This is a milestone in the long history
in statistical physics and combinatorial algorithms.
But the case over general graphs is different.
In 1979
L. Valiant \cite{Valiant-permanent-paper79-TCS} defined the class
\#P for counting problems. Most counting problems of a combinatorial
nature,
 Sum-of-Product computations
such as partition functions studied in physics, and counting
 Constraint Satisfaction Problems (\#CSP) are all included in \#P
(or more precisely in
$\operatorname{FP}^{\operatorname{\#P}}$ as the output may be
non-integers). In particular,
counting perfect matchings in general graphs is \#P-complete~\cite{valiant-enum}.

In two seminal papers~\cite{Val02a,string23},
L.~Valiant introduced \emph{matchgates} and \emph{holographic algorithms}.
Computation in holographic algorithms based on matchgates is expressed and interpreted through a choice of linear basis vectors in an exponential ``holographic" mix. Then the actual computation is carried out, via the Holant Theorem, by
the FKT algorithm  for counting the number of perfect matchings in a planar graph.
 This methodology has produced polynomial time algorithms for a variety of problems ranging from restrictive versions of Satisfiability, Vertex Cover, to other graph problems such as edge orientation and node/edge deletion. No polynomial time algorithms were known for any of these problems, and some minor variations are known to be NP-hard.

In the past decade significant progress was made
in the understanding of these remarkable
 algorithms~\cite{Cai-Fu-Guo-W, caiguowilliams13,  art-sc,
Cai-Lu-Xia-real, Guo-Williams,
 landsberg-morton, Val02b, string23, Val06}.
In an interesting twist, it turns out
that the idea of a holographic reduction is not only
a powerful technique to design unexpected algorithms
(tractability),
but also  an indispensable tool to prove intractability
and then
to prove classification theorems.
Furthermore, in  a  self-referential twist,
it has proved to be
a crucial tool to
understand the limit and scope of the newly introduced holographic
algorithms themselves~\cite{clx-holant, ghlx-stacs-2011, HL12, glva-2013,
caiguowilliams13, Guo-Williams, Cai-Fu-Guo-W}.
This study has produced
a number of complexity dichotomy theorems.
These classify \emph{every} problem expressible
 in a framework as either solvable in P-time or being \#P-hard,
with nothing in between.

One such framework is called (weighted) \#CSP problems.
Let $[q] = \{0, 1, \cdots, q-1\}$ denote a domain of size $q$.
A \#CSP problem over the domain $[q]$
  (for $q=2$,  it is the Boolean domain)
is specified  by a fixed finite set $\mathcal{F}$ of local constraint
functions. Each function $f \in \mathcal{F}$ has an arity $k$,
and maps $[q]^k \rightarrow \mathbb{C}$.
(Unweighted \#CSP problems are defined by 0-1 valued constraint
functions.)
An instance of \#CSP($\mathcal{F}$) is specified by
a finite set of variables $X = \{x_1, x_2, \ldots, x_n\}$,
and a finite sequence of constraints $\mathcal{S}$ from $\mathcal{F}$, each
applied to  an ordered sequence of variables from $X$.
The output of this instance is $\sum_{\sigma}
 \prod_{f \in \mathcal{S}} f|_{\sigma}$, a sum over all
$\sigma: X \rightarrow [q]$, of products
of all constraints in $\mathcal{S}$ evaluated according to $\sigma$.
\#CSP is a very expressive framework for locally specified counting
problems. E.g., all spin systems are special cases  where
$\mathcal{F}$ consists of a single binary constraint,
and possibly some unary constraints when there are ``external fields''.

The following classification theorem
 for \#CSP on the Boolean domain  is proved in
gradually increasing
generalities~\cite{Cai-Lu-Xia-real, Guo-Williams, cai-fu-plcsp},
  reaching full generality
in \cite{cai-fu-plcsp}:
\begin{theorem}\label{theorem-main-intro}
For any  finite
  set of constraint functions $\mathcal{F}$ over Boolean variables,
each  taking (algebraic) complex values and not necessarily symmetric,
\#CSP($\mathcal{F}$) belongs to exactly one of three categories
according to $\mathcal{F}$:
(1) It is P-time solvable;
(2) It is P-time solvable over planar graphs but \#P-hard over general graphs;
(3) It is \#P-hard over planar graphs.
Moreover,
category (2) consists \emph{precisely} of those problems that are holographically reducible to the FKT algorithm,
whereby all constraint functions in  $\mathcal{F}$ and
 {\sc Equality} functions
of all arities are transformed to matchgate signatures.
\end{theorem}
Theorem~\ref{theorem-main-intro} shows that
 holographic algorithms with matchgates
form a \emph{universal} strategy, that applies a holographic transformation whereby we transform all {\sc Equality} functions
to matchgates,  for all problems in this framework
 that are \#P-hard in general but solvable
in polynomial time on planar graphs.
But for \#CSP over higher domains, the situation is different.
Over the general domain of size $q \ge 3$,
there are only a few holographic algorithms
with matchgates~\cite{string25,Cai-Fu}.  But
it can be argued that they are problems that actually get transformed to a
Boolean domain \#CSP problems.
In the present paper, we prove the following theorem:
\begin{theorem}\label{equality-no-matchgate}
For any $n\geq 3$ and $q\geq 3$,
there is no $q\times 2^{\ell}$ matrix $M$ of rank $q$ such that
the transformed  {\sc Equality}  function  $(=_n)M^{\otimes n}$
by $M$
can be realized by a matchgate signature, where $(=_n)$
 is the {\sc Equality} function on $n$ variables over the domain $[q]$.
\end{theorem}
This result has the consequence that
the standard strategy  that is universal for \#CSP on the Boolean domain, cannot work in any higher domain
$[q]$, where $q\geq 3$.
In the proof of Theorem~\ref{equality-no-matchgate}, a main tool is
the use of Matchgate Identities (MGI).
These are a set of polynomial equations on the
entries of a matchgate realizable constraint function (signature)
and they have deep relations to some branches of mathematics and complexity theory.
In \cite{c-l-collapse, sitan, mingji}, MGI were used to prove the  basis collapse theorems and many
techniques were developed.
In the present paper, we contribute some new techniques to use MGI.

Theorem~\ref{equality-no-matchgate} is a part of our effort
to achieve a full structural characterization of
 blockwise symmetric matchgate signatures.
This improves upon the work by Cai and Lu~\cite{cailu10}.
Our characterization theorem  will be stated in Theorem~\ref{fu'stheorem1},
after we have defined all the necessary terminology.
Compared to the theorem in \cite{cailu10} we eliminate a
non-vanishing condition and improve the characterization
 to all arity $n \ge 3$.
This improvement makes  Theorem~\ref{equality-no-matchgate}
above possible.
We also give a counter example to indicate that Theorem~\ref{fu'stheorem1}
cannot be further improved to $n=2$.
So we have achieved  a complete characterization of the
structure of blockwise symmetric matchgate signatures.

Beyond \#CSP,
there is a framework called Holant problems for locally specified
counting problems, of  which  \#CSP is a special case.
In \cite{Cai-Fu-Guo-W}, it is proved that, the analogous universality statement
for Holant problems over the Boolean domain is \emph{not true}:
In addition to holographic reductions to matchgates,
there is another set of problems (expressible in the Holant
framework but not in \#CSP) that is \#P-hard in general but P-time tractable
over planar graphs. Moreover they cannot be transformed to
matchgates, and the P-time algorithm is not the type discussed
in this paper.  Thus the present paper also highlights the
distinction and importance of Holant problems, beyond
\#CSP. It is an open question how to classify
 Holant problems for higher domains, especially those
problems that are  \#P-hard in general but P-time tractable
over planar graphs.

\section{Preliminaries}
In the following, we always use
$e_i$ to denote the string with 1 in the $i$-th bit and 0 elsewhere (of some length $n$)and use $\mathbf{0}$ to denote the $\ell$-bit string $00\cdots 0$.

Let $G=(V, E, W)$ be a weighted undirected planar graph. A {\it matchgate} $\Gamma$ is a tuple $(G, X)$ where $X\subset V$ is a set of external
nodes, ordered counterclockwise on the external face. $\Gamma$ is called an odd (resp. even) matchgate if it has an odd (resp. even)
number of nodes.

Each matchgate $\Gamma$ with $n$ external nodes is assigned a {\it matchgate signature} $(\Gamma^{\alpha})_{\alpha\in\{0, 1\}^{n}}$ with $2^{n}$ entries,
\begin{equation*}
\Gamma^{i_{1}i_{2}\cdots i_{n}}=\textrm{PerfMatch}(G-Z)=\displaystyle\sum_{M}\prod_{(i, j)\in M}w_{ij},
\end{equation*}
where the sum is over all perfect matchings $M$ of $G-Z$, and $Z\subset X$ is a subset of external nodes having the characteristic
sequence $\chi_{Z}=i_{1}i_{2}\cdots i_{n}$ and $G-Z$ is obtained from $G$ by removing $Z$ and its incident edges.

An entry $\Gamma^{\alpha}$ is called an even (resp. odd) entry if the Hamming weight wt$(\alpha)$ is even (resp. odd). It was proved in \cite{Matchgates-Reconsidered}
that matchgate signatures are characterized by the following two sets of conditions. (1) The parity requirements: either all
even entries are 0 or all odd entries are 0. This is due to perfect matchings. (2)
A set of Matchgate Identities (MGI) defined as
follows:
For any $\alpha\in\{0, 1\}^n$ and any position vector $P=\{p_1, p_2, \cdots, p_{\ell}\}$, where $p_{1}<p_{2}< \cdots<p_{\ell}$,
also denoted by a bit string,
%
\begin{equation}
\displaystyle\sum_{i=1}^{\ell}(-1)^{i}\Gamma^{\alpha+e_{p_{i}}}\Gamma^{\alpha+P+e_{p_{i}}}=0
\end{equation}
(alternating sum by flipping in sequence the bits $p_i$
and the bits in $P\setminus \{p_i\}$), where $\alpha+\beta$ denotes the XOR of $\alpha$ and $\beta$.
Actually in \cite{Matchgates-Reconsidered} it is shown that
MGI implies the Parity Condition. But in practice, it is easier
to apply the Parity Condition first.
We use $\mathscr{M}$ to denote the set of matchgate signatures.

Now we introduce Holant problems over domain $[q]$. Fix a set $\mathcal{F}$ of  local constraint
functions, a.k.a. signatures. 
A \emph{signature grid} $\Omega = (G, \pi)$ consists of a graph $G = (V,E)$,
and a mapping $\pi$ which maps each vertex $v \in V$ to
some $f_v \in \mathcal{F}$ on domain $\{0, 1, \cdots, q-1\}$ of arity $\deg(v)$ ($=$ the
number of input variables of $f_v$),
and associates its incident edges $E(v)$ to the input variables of $f_v$.
We say that $\Omega$ is a \emph{planar signature grid} if $G$ is a plane graph,
where the variables of $f_v$ are ordered counterclockwise starting from an edge specified by $\pi$.
The Holant problem on instance $\Omega$ is to evaluate
\[\operatorname{Holant}(\Omega; \mathcal{F}) = \sum_{\sigma: E \to \{0,1, \cdots, q-1\}}
\prod_{v \in V} f_v(\sigma \mid_{E(v)}),\]
where $\sigma \mid_{E(v)}$ denotes the restriction of $\sigma$ to $E(v)$.
A Holant problem is parameterized by a set of signatures.
 Given a set of signatures $\mathcal{F}$,
  the counting problem Holant$(\mathcal{F})$ is as follows:
The input is a \emph{signature grid} $\Omega = (G, \pi)$; the
 output is Holant$(\Omega; \mathcal{F})$.
The problem Pl-Holant$(\mathcal{F})$ is defined similarly using a planar signature grid.
For example, for domain size $q=2$, if we place the {\sc Exact-One} function at every vertex,
then Pl-Holant$(\Omega; \mathcal{F})$ counts the
number of {\sc Perfect Matchings}. Note that this problem can be computed in polynomial time by
the Kasteleyn's algorithm (a.k.a. the  FKT algorithm) \cite{Kasteleyn1961, Kasteleyn1967, TF61}.
Moreover, if $\mathcal{F}\subseteq\mathscr{M}$, then Pl-Holant$(\mathcal{F})$ can be computed in polynomial time.
We use $(=_n)$ to denote the $n$-ary {\sc Equality} signature, i.e., the signature takes 1 on the input $x_1=x_2=\cdots=x_n$ and 0 on other inputs.

To define holographic reductions,
we use Holant$(\mathcal{F}|\mathcal{G})$ to denote the Holant problem over signature grids with a bipartite graph $G= (U,V,E)$,
where each vertex in $U$ or $V$ is assigned a signature in $\mathcal{F}$ or $\mathcal{G}$,
respectively.
Signatures in $\mathcal{F}$ are given as
row vectors listing the function values like
truth tables;
 signatures in $\mathcal{G}$ are given as column vectors.
We use
Pl-Holant$(\mathcal{F}|\mathcal{G})$ to denote
 the Holant problem over planar bipartite graphs.
Counting constraint satisfaction problems (\#CSP)
can be defined as a special case of Holant problems.
An instance of \#CSP$(\mathcal{F})$ is presented
as a bipartite graph.
There is one node for each variable on LHS and for each occurrence
of constraint functions on RHS respectively.
Connect a constraint node to  a variable node if the
variable appears in that occurrence
of constraint, with a labeling on the edges
for the order of these variables.
This bipartite graph is also known as the \emph{constraint graph}.
If we attach each variable node with an {\sc Equality} signature,
and consider every edge as a variable, then
the \#CSP problem is just the Holant problem on this bipartite graph.
Thus
\#CSP$(\mathcal{F}) \equiv_T \operatorname{Holant}(\mathcal{EQ}|\mathcal{F})$,
where $\mathcal{EQ} = \{{=}_1, {=}_2, {=}_3, \dotsc\}$ is the set of \textsc{Equality} signatures of all arities.
By restricting to planar constraint graphs,
we have Pl-\#CSP.

Now we define holographic transformations.
Let $M$ be a $q\times 2^{\ell}$ matrix of rank $q$, then there exists a $2^{\ell}\times q$ matrix $\check{M}$ of rank $q$
such that $M\check{M}=I_{q}$.
%
For a signature $f$ of arity $n$ on domain $\{0, 1, \cdots, q-1\}$, written as
a column vector  $f \in \mathbb{C}^{q^n}$, we denote by
$\check{M}f = (\check{M})^{\otimes n} f$ the transformed signature.
  For a signature set $\mathcal{F}$,
define $\check{M} \mathcal{F} = \{\check{M}f \mid  f \in \mathcal{F}\}$.
For signatures written as
 row vectors we define $\mathcal{F} M$ similarly.
In particular, for $q=2$, let
$H_2 = \frac{1}{\sqrt{2}} \left[\begin{smallmatrix} 1 & 1 \\ 1 & -1 \end{smallmatrix}\right]$
be the Hadamard matrix and we denote
 $\widehat{\mathcal{F}} = H_2  \mathcal{F}$.
Note that $(H_2)^{-1}=H_2$ and $\mathcal{EQ}\subseteq\widehat{\mathscr{M}}$.
This fact, that under the Hadamard transformation all {\sc Equalities'} become matchgate signatures is proved to be the universal reason
that \#CSP
problems that are \#P-hard can become P-time tractable on planar structures. This is proved in
\cite{cai-fu-plcsp}.

The holographic transformation defined by $M$ is the following operation:
given a signature grid $\Omega = (G, \pi)$ of Holant$(\mathcal{F}|\mathcal{G})$,
for the same bipartite graph $G$,
we get a new signature grid $\Omega'= (G, \pi')$ of Holant$(\mathcal{F} M|\check{M} \mathcal{G})$ by replacing each signature in
$\mathcal{F}$ or $\mathcal{G}$ with the corresponding signature in $\mathcal{F} M$ or $\check{M} \mathcal{G}$.
Then we have the following theorem.
\begin{theorem}\label{valiant}\rm {(Valiant's Holant Theorem\cite{string23})}.
Let $\mathcal{F}, \mathcal{G}$ be signature sets on domain $\{0, 1, \cdots, q-1\}$ and $M$ be a $q\times 2^{\ell}$ matrix of rank $q$, then
\begin{center}
$\operatorname{Holant}(\mathcal{F}|\mathcal{G})=\operatorname{Holant}(\mathcal{F} M|\check{M} \mathcal{G})$.
\end{center}
\end{theorem}
By Theorem~\ref{valiant} and $\mathcal{EQ}\subseteq\widehat{\mathscr{M}}$, if $\mathcal{F}\subseteq\widehat{\mathscr{M}}$, then
Pl-\#CSP$(\mathcal{F})$ can be computed in polynomial time by the holographic transformation using $H_2$ and the FKT algorithm.
More precisely, we have the following complexity trichotomy
 theorem
 for \#CSP on the Boolean domain \cite{Cai-Lu-Xia-real, Guo-Williams, cai-fu-plcsp}:
{\it every single problem} in this class can be classified into one of three types.
The first type  can be
solved in polynomial time over arbitrary structures.
The second type consists of problems that are
 \#P-hard over general structures, but solvable in
polynomial time over planar structures.
The third type problems are those which remain \#P-hard
over planar structures.
In this trichotomy
 theorem, the second type of problems are precisely captured by $\widehat{\mathscr{M}}$.

\begin{definition}{\rm
  A signature $f=(f^{i_1i_2\cdots i_n})$, where each $i_j\in[q]$, is {\it symmetric} if $f$ is invariant under
  any permutation of $\{i_1, i_2, \cdots, i_n\}$ i.e.,
$f^{ i_{1}i_2\cdots i_{n}}=
f^{i_{\sigma(1)}i_{\sigma(2)}\cdots i_{\sigma(n)}}$
for any permutation $\sigma$ of  $\{1, 2 \cdots, n\}$.}
\end{definition}
For example, the {\sc Equality} signatures are symmetric.

\begin{definition}
{\rm
For an $n$-ary symmetric signature $f$ over domain $[q]$,
its matrix form $M(f)$
 is a $q\times q^{n-1}$ matrix. Its rows are
indexed by  $i_{1}$ and its columns are indexed
by $i_{2}i_{3}\cdots i_{n}$.}
\end{definition}

For a string
$\alpha=(i_1i_2\cdots i_{\ell})(i_{\ell+1}i_{\ell+2}\cdots i_{2\ell})\cdots(i_{(n-1)\ell+1}i_{(n-1)\ell+2}\cdots i_{n\ell})\in\{0, 1\}^{n\ell}$,
 we call $\alpha_j=i_{(j-1)\ell+1}i_{(j-1)\ell+2}\cdots i_{j\ell}$, $1\leq j\leq n$, blocks in size-$\ell$
 and
 call $\alpha=\alpha_1\alpha_2\cdots\alpha_n$
 the blockwise form  of $\alpha$ in size-$\ell$.

\begin{definition}{\rm
An $n\ell$-ary signature $\Gamma$ over the Boolean domain is blockwise symmetric in size-$\ell$ if
$\Gamma=(\Gamma^{\alpha_1\alpha_2\cdots\alpha_n})$, where each $\alpha_{i}\in\{0, 1\}^{\ell}$,  is
  invariant under any permutation of $\{\alpha_1, \alpha_2, \cdots, \alpha_n\}$ i.e.,
$\Gamma^{\alpha_{1}\alpha_2\cdots\alpha_{n}}=
\Gamma^{\alpha_{\sigma(1)}\alpha_{\sigma(2)}\cdots\alpha_{\sigma(n)}}$
for any permutation $\sigma$ of  $\{1, 2 \cdots, n\}$.}
In particular, if $\ell=1$, we say $\Gamma$ is bitwise symmetric.
\end{definition}
If it is clear in the context, we will omit size-$\ell$.

\begin{definition}
{\rm
The {\it matrix form} $M(\Gamma)$
of the blockwise symmetric signature $\Gamma=(\Gamma^{\alpha_{1}\alpha_{2}\cdots \alpha_{n}})$ of arity $n\ell$
 is a $2^{\ell}\times 2^{(n-1)\ell}$ matrix. Its rows are
indexed by  $\alpha_{1}$ and its columns are indexed
by $\alpha_{1}\alpha_{2}\cdots \alpha_{n}$.}
\end{definition}
For example, let $\Gamma=(\Gamma^{\alpha_{1}\alpha_{2}})$, where $n=2, \ell=2$, then
$M(\Gamma)=\left(\begin{smallmatrix}
\Gamma^{0000}& \Gamma^{0001}& \Gamma^{0010}& \Gamma^{0011}\\
\Gamma^{0100}& \Gamma^{0101}& \Gamma^{0110}& \Gamma^{0111}\\
\Gamma^{1000}& \Gamma^{1001}& \Gamma^{1010}& \Gamma^{1011}\\
\Gamma^{1100}& \Gamma^{1101}& \Gamma^{1110}& \Gamma^{1111}
\end{smallmatrix}\right)$.
We denote the $\alpha$-th row of $M(\Gamma)$ by $M(\Gamma)^{\alpha}$ for $\alpha\in\{0, 1\}^{\ell}$.

The following is a simple lemma from Linear Algebra.
\begin{lemma}\label{rank}
Let $A, B, C$ be $m\times n, n\times s, s\times t$ matrices respectively, where rank$(A)=n$, rank$(C)=s$,
then rank$(AB)=$rank$(B)$, rank$(BC)=$rank$(B)$.
\end{lemma}

\begin{lemma}\label{symmetry}
Let $fM^{\otimes n}=\Gamma$, where $f$ is an n-ary signature on domain $[q]$,
$M$ is a $q\times 2^{\ell}$ matrix and $\Gamma=(\Gamma^{\alpha_1\alpha_2\cdots\alpha_n})$ is an $n\ell$-ary signature where
$\alpha_i\in\{0, 1\}^{\ell}$ for $1\leq i\leq n$.
If $f$ is symmetric, then $\Gamma$ is blockwise symmetric.
\end{lemma}
\begin{proof}
Let $M=(M_i^{\alpha})$, where $i\in [n]$ is the index of rows and $\alpha\in\{0, 1\}^{\ell}$
is the index of columns.
From $\Gamma=fM^{\otimes n}$, we have
$\Gamma^{\alpha_1\alpha_2\cdots\alpha_n}=\displaystyle\sum_{i_1, i_2, \cdots, i_n\in[q]}
f^{i_1i_2\cdots i_n}M_{i_1}^{\alpha_1}M_{i_2}^{\alpha_2}\cdots M_{i_n}^{\alpha_n}.$
Then for any permutation $\sigma$ of $\{1, 2, \cdots, n\}$,  we have
\[\Gamma^{\alpha_{\sigma(1)}\alpha_{\sigma(2)}\cdots\alpha_{\sigma(n)}}=\displaystyle\sum_{i_{\sigma(1)}, i_{\sigma(2)}, \cdots, i_{\sigma(n)}\in[q]}
f^{i_{\sigma(1)}i_{\sigma(2)}\cdots i_{\sigma(n)}}M_{i_{\sigma(1)}}^{\alpha_{\sigma(1)}}M_{i_{\sigma(2)}}^{\alpha_{\sigma(2)}}\cdots M_{i_{\sigma(n)}}^{\alpha_{\sigma(n)}}.\]
Note that $f^{i_{\sigma(1)}i_{\sigma(2)}\cdots i_{\sigma(n)}}=f^{i_1i_2\cdots i_n}$ since $f$ is symmetric, and
$M_{i_1}^{\alpha_1}M_{i_2}^{\alpha_2}\cdots M_{i_n}^{\alpha_n}=M_{i_{\sigma(1)}}^{\alpha_{\sigma(1)}}M_{i_{\sigma(2)}}^{\alpha_{\sigma(2)}}\cdots M_{i_{\sigma(n)}}^{\alpha_{\sigma(n)}}$
since $\{M_{i_{\sigma(1)}}^{\alpha_{\sigma(1)}}, M_{i_{\sigma(2)}}^{\alpha_{\sigma(2)}}, \cdots,  M_{i_{\sigma(n)}}^{\alpha_{\sigma(n)}}\}$
is just a permutation of $\{M_{i_1}^{\alpha_1}, M_{i_2}^{\alpha_2}, \cdots,  M_{i_n}^{\alpha_n}\}$.
Thus $\Gamma^{\alpha_1\alpha_2\cdots\alpha_n}=\Gamma^{\alpha_{\sigma(1)}\alpha_{\sigma(2)}\cdots\alpha_{\sigma(n)}}.$
This implies that $\Gamma$ is blockwise symmetric.
\end{proof}

If  $fM^{\otimes n}=\Gamma$, where $f$ is an $n$-ary symmetric signature on domain $[q]$ and $M$ is a $q\times 2^{\ell}$ matrix, then
$\Gamma=(\Gamma^{\alpha_1\alpha_2\cdots\alpha_n})$, where
$\alpha_i\in\{0, 1\}^{\ell}$ for $1\leq i\leq n$, is a
blockwise symmetric signature by Lemma~\ref{symmetry}. Then we have the following lemma.

\begin{lemma}(Lemma~2.2 of \cite{Cai-Fu})\label{signature-matrix}
 $M(\Gamma)=M^T M(f)M^{\otimes (n-1)}$, where $M^T$ is the transpose of $M$.
\end{lemma}


\section{Equality $(=_n)$ with $n\geq 3$ on domain size $\geq 3$ cannot be realized by matchgates}\label{main-section}
Let $\Gamma=(\Gamma^{\alpha_{1}\alpha_{2}\cdots \alpha_{n}})$, where
$\alpha_i\in\{0, 1\}^{\ell}$ for $1\leq i\leq n$,  be a blockwise symmetric matchgate signature
and $M(\Gamma)$ be the matrix form of $\Gamma$. If rank$(M(\Gamma))\geq 2$, then there exist
$\sigma, \tau\in\{0, 1\}^{\ell}$ that satisfy the following conditions:
\begin{itemize}
\item $M(\Gamma)^{\sigma}$ and  $M(\Gamma)^{\tau}$ are linearly independent,
\item ${\rm wt}(\sigma+\tau)=\displaystyle\min_{ u, v\in\{0, 1\}^{\ell}}\{{\rm wt}(u+ v) \mid M(\Gamma)^{u}$ and $M(\Gamma)^{v}$
  are linearly independent$\}$,
\end{itemize}
where $\sigma+\tau$ is the XOR of the bit strings $\sigma$ and $\tau$.
Moreover, For any
  $\beta=\alpha_1\cdots\alpha_{t-1}\alpha_{t+1}\cdots\alpha_n\in\{0,1\}^{(n-1)\ell}$, let $x_{\beta}=\begin{pmatrix}
\Gamma^{\alpha_1\cdots\alpha_{t-1}\sigma\alpha_{t+1}\cdots\alpha_n}\\
\Gamma^{\alpha_1\cdots\alpha_{t-1}\tau\alpha_{t+1}\cdots\alpha_n}
\end{pmatrix}$. Then there exist $\zeta, \eta$ satisfy the following conditions:

\begin{itemize}
\item $x_{\zeta}$ and $x_{\eta}$ are linearly independent,
\item ${\rm wt}(\zeta+ \eta)=\displaystyle\min_{ u, v\in\{0, 1\}^{(n-1)\ell}}\{{\rm wt}(u+ v) \mid x_{u}$  and $x_{v}$ are linearly independent$\}$.
\end{itemize}
For such $\sigma, \tau, \zeta, \eta$, 
 the following lemma is given in \cite{Cai-Fu}, which uses the properties of Matchgate (MGI).
\begin{lemma}\label{le-col-2-1}
If {\rm rank}$(M(\Gamma))\geq 2$, then ${\rm wt}(\sigma+\tau)=1$ and ${\rm wt}(\zeta+\eta)=1$.
\end{lemma}
For a blockwise symmetric matchgate signature $\Gamma=(\Gamma^{\alpha_{1}\alpha_{2}\cdots \alpha_{n}})$ of arity $n\ell$,
by
Lemma \ref{le-col-2-1} and the parity condition,
if rank$(M(\Gamma))\geq 2$,
then $M(\Gamma)$ has a full rank submatrix of the form (zeros are due to Parity Constraint)
\begin{center}
$\left(\begin{smallmatrix}
\Gamma^{\alpha_{1}\alpha_{2}\cdots\alpha_{n}}&0\\
0&\Gamma^{(\alpha_{1}+e_s)(\alpha_{2}+e_t)\cdots\alpha_{n}}
\end{smallmatrix}\right)~~~~~~$
or
$~~~~~~\left(\begin{smallmatrix}
0&\Gamma^{\alpha_{1}\alpha_{2}\cdots\alpha_{n}}\\
\Gamma^{(\alpha_{1}+e_s)(\alpha_{2}+e_t)\cdots\alpha_{n}}&0
\end{smallmatrix}\right), $
\end{center}
where $e_s, e_t\in\{0, 1\}^{\ell}$.

Moreover, let $\Gamma=(\Gamma^{\alpha_{1}\alpha_{2}\cdots \alpha_{n}})$ be a blockwise symmetric matchgate signature of arity $n\ell$.
If rank$(M(\Gamma))\geq 3$, then there exist
$\sigma, \tau\in\{0, 1\}^{\ell}$ satisfy the following conditions:
\begin{itemize}
\item wt$(\sigma)$ and wt$(\tau)$ have the same parity,
\item $M(\Gamma)^{\sigma}$ and  $M(\Gamma)^{\tau}$ are linearly independent,
\item ${\rm wt}(\sigma+\tau)=\displaystyle\min_{ u, v\in\{0, 1\}^{\ell}}\{{\rm wt}(u+ v) \mid$
wt$(u)$ and wt$(v)$ have the same parity and  $M(\Gamma)^{u}$, $M(\Gamma)^{v}$ are linearly independent$\}$.
\end{itemize}
Moreover, For any
  $\beta=\alpha_1\cdots\alpha_{t-1}\alpha_{t+1}\cdots\alpha_n\in\{0,1\}^{(n-1)\ell}$, let $x_{\beta}=\begin{pmatrix}
\Gamma^{\alpha_1\cdots\alpha_{t-1}\sigma\alpha_{t+1}\cdots\alpha_n}\\
\Gamma^{\alpha_1\cdots\alpha_{t-1}\tau\alpha_{t+1}\cdots\alpha_n}
\end{pmatrix}$. Then there exist $\zeta, \eta$ satisfy the following condition:
\begin{itemize}
\item $x_{\zeta}$ and $x_{\eta}$ are linearly independent,
\item ${\rm wt}(\zeta+ \eta)=\displaystyle\min_{ u, v\in\{0, 1\}^{(n-1)\ell}}\{{\rm wt}(u+ v) \mid x_{u}$ and $x_{v}$
are linearly independent$\}$.
\end{itemize}
For such $\sigma, \tau, \zeta, \eta$, 
the following Lemma is given in \cite{Cai-Fu} (again using MGI),
\begin{lemma}\label{le-col-3-1}
If rank$(M(\Gamma))\geq 3$, then
${\rm wt}(\sigma+\tau)=2$,
${\rm wt}(\zeta+\eta)=2$.
  \end{lemma}

For a blockwise symmetric matchgate signature $\Gamma=(\Gamma^{\alpha_{1}\alpha_{2}\cdots \alpha_{n}})$ of arity $n\ell$,
 Lemma \ref{le-col-3-1} implies that if rank$(M(\Gamma))\geq 3$, then $M(\Gamma)$ has a full rank submatrix of the form
\begin{equation}\left(\begin{smallmatrix}\label{rank3-1}
\Gamma^{\alpha_{1}\alpha_{2}\alpha_3\cdots\alpha_{n}}&\Gamma^{\alpha_{1}(\alpha_{2}+e_{s}+e_{t})\alpha_{3}\cdots\alpha_{n}}\\
\Gamma^{(\alpha_{1}+e_{i}+e_{j})\alpha_2\alpha_3\cdots\alpha_{n}}&\Gamma^{(\alpha_{1}+e_{i}+e_{j})(\alpha_{2}+e_{s}+e_{t})\alpha_3\cdots\alpha_{n}}
\end{smallmatrix}\right)\end{equation}
or
\begin{equation}\label{rank3-2}\left(\begin{smallmatrix}
\Gamma^{\alpha_{1}\alpha_{2}\alpha_3\cdots\alpha_{n}}&\Gamma^{\alpha_{1}(\alpha_{2}+e_{s})(\alpha_{3}+e_{t})\cdots\alpha_{n}}\\
\Gamma^{(\alpha_{1}+e_{i}+e_{j})\alpha_2\alpha_3\cdots\alpha_{n}}&\Gamma^{(\alpha_{1}+e_{i}+e_{j})(\alpha_{2}+e_{s})(\alpha_{3}+e_{t})\cdots\alpha_{n}}
\end{smallmatrix}\right),\end{equation}
where
$e_i, e_j, e_s, e_t$ are in $\{0, 1\}^{\ell}$, $s< t$ in $(\ref{rank3-1})$  and $i< j$.
But the next two lemmas show that for
any $\{\beta_1, \beta_2, \cdots, \beta_n, \gamma_1, \gamma_2, \cdots, \gamma_n\},$ where $\beta_i, \gamma_i\in\{0, 1\}^{\ell}$,
such that wt$(\beta_1+\gamma_1)=2$, wt$(\beta_{2}\beta_{3}\cdots\beta_{n}+\gamma_{2}\gamma_{3}\cdots\gamma_{n})=2$,
the submatrix of $M(\Gamma)$
$\left(
\begin{smallmatrix}
\Gamma^{\beta_{1}\beta_{2}\beta_{3}\cdots\beta_{n}}&\Gamma^{\beta_{1}\gamma_{2}\gamma_{3}\cdots\gamma_{n}}\\
\Gamma^{\gamma_1\beta_{2}\beta_{3}\cdots\beta_{n}}&\Gamma^{\gamma_1\gamma_{2}\gamma_{3}\cdots\gamma_{n}}
\end{smallmatrix}
\right)
$
has rank less than 2. 

\begin{lemma}\label{xxx-1}
Let $\Gamma=(\Gamma^{\alpha_{1}\alpha_{2}\cdots\alpha_{n}})$ be a blockwise symmetric matchgate signature of arity $n\ell$, where $n\geq 3$.
Then for any $\alpha_{1}\alpha_{2}\cdots\alpha_{n}\in\{0, 1\}^{n\ell}$,
\begin{equation*}
A=\left(\begin{smallmatrix}
\Gamma^{\alpha_{1}\alpha_{2}\alpha_{3}\cdots\alpha_{n}}&\Gamma^{\alpha_{1}(\alpha_{2}+e_{s})(\alpha_{3}+e_{t})\cdots\alpha_{n}}\\
\Gamma^{(\alpha_{1}+e_{i}+e_{j})\alpha_{2}\alpha_{3}\cdots\alpha_{n}}&\Gamma^{(\alpha_{1}+e_{i}+e_{j})(\alpha_{2}+e_{s})(\alpha_{3}+e_{t})\cdots\alpha_{n}}
\end{smallmatrix}
\right)
\end{equation*}
is degenerate, where
$e_i, e_j, e_s, e_t\in\{0, 1\}^{\ell}$ and $i<j$.
\end{lemma}
\begin{proof}
We will prove that $\det(A)=0$ by applying Matchgate Identities (MGI).

 Let the pattern be $(\alpha_{1}+e_{i})\alpha_{2}\alpha_{3}\cdots\alpha_{n}$
and the position vector be $(e_{i}+e_{j})(e_{s})(e_{t})\mathbf{0}\cdots \mathbf{0}$. Then by MGI we have
\begin{equation}\label{MGI1}
\begin{split}
&\Gamma^{\alpha_{1}\alpha_{2}\alpha_{3}\cdots\alpha_{n}}\Gamma^{(\alpha_{1}+e_{i}+e_{j})(\alpha_{2}+e_{s})(\alpha_{3}+e_{t})\cdots\alpha_{n}}
-\Gamma^{(\alpha_{1}+e_{i}+e_{j})\alpha_{2}\alpha_{3}\cdots\alpha_{n}}\Gamma^{\alpha_{1}(\alpha_{2}+e_{s})(\alpha_{3}+e_{t})\cdots\alpha_{n}}\\
+&\Gamma^{(\alpha_{1}+e_{i})(\alpha_{2}+e_{s})\alpha_{3}\cdots\alpha_{n}}\Gamma^{(\alpha_{1}+e_{j})\alpha_{2}(\alpha_{3}+e_{t})
\cdots\alpha_{n}}
-\Gamma^{(\alpha_{1}+e_{i})\alpha_{2}(\alpha_{3}+e_{t})\cdots\alpha_{n}}\Gamma^{(\alpha_{1}+e_{j})(\alpha_{2}+e_{s})\alpha_{3}\cdots\alpha_{n}}=0.
\end{split}
\end{equation}
This instantiation of MGI has 4 terms, each component to flipping the bits at position $i, j, s, t$ respectively.

Moreover, let the pattern be $(\alpha_{1}+e_{i})\alpha_{3}\alpha_{2}\alpha_{4}\cdots\alpha_{n}$
and the position vector be $(e_{i}+e_{j})(e_{t})(e_{s})\mathbf{0}\cdots \mathbf{0}$. Then we have
\begin{equation}\label{MGI2}
\begin{split}
&\Gamma^{\alpha_{1}\alpha_{3}\alpha_{2}\cdots\alpha_{n}}\Gamma^{(\alpha_{1}+e_{i}+e_{j})(\alpha_{3}+e_{t})(\alpha_{2}+e_{s})\cdots\alpha_{n}}
-\Gamma^{(\alpha_{1}+e_{i}+e_{j})\alpha_{3}\alpha_{2}\cdots\alpha_{n}}\Gamma^{\alpha_{1}(\alpha_{3}+e_{t})(\alpha_{2}+e_{s})\cdots\alpha_{n}}\\
+&\Gamma^{(\alpha_{1}+e_{i})(\alpha_{3}+e_{t})\alpha_{2}\cdots\alpha_{n}}\Gamma^{(\alpha_{1}+e_{j})\alpha_{3}(\alpha_{2}+e_{s})
\cdots\alpha_{n}}
-\Gamma^{(\alpha_{1}+e_{i})\alpha_{3}(\alpha_{2}+e_{s})\cdots\alpha_{n}}\Gamma^{(\alpha_{1}+e_{j})(\alpha_{3}+e_{t})\alpha_{2}\cdots\alpha_{n}}=0.
\end{split}
\end{equation}
Since $\Gamma$ is blockwise symmetric, (\ref{MGI2}) can be rewritten as the following form:
\begin{equation}\label{MGI2'}
\begin{split}
&\Gamma^{\alpha_{1}\alpha_{2}\alpha_{3}\cdots\alpha_{n}}\Gamma^{(\alpha_{1}+e_{i}+e_{j})(\alpha_{2}+e_{s})(\alpha_{3}+e_{t})\cdots\alpha_{n}}
-\Gamma^{(\alpha_{1}+e_{i}+e_{j})\alpha_{2}\alpha_{3}\cdots\alpha_{n}}\Gamma^{\alpha_{1}(\alpha_{2}+e_{s})(\alpha_{3}+e_{t})\cdots\alpha_{n}}\\
+&\Gamma^{(\alpha_{1}+e_{i})\alpha_{2}(\alpha_{3}+e_{t})\cdots\alpha_{n}}\Gamma^{(\alpha_{1}+e_{j})(\alpha_{2}+e_{s})\alpha_{3}
\cdots\alpha_{n}}
-\Gamma^{(\alpha_{1}+e_{i})(\alpha_{2}+e_{s})\alpha_{3}\cdots\alpha_{n}}\Gamma^{(\alpha_{1}+e_{j})\alpha_{2}(\alpha_{3}+e_{t})\cdots\alpha_{n}}=0.
\end{split}
\end{equation}
Note that the $1, 2$-th terms of (\ref{MGI1}) are equal to the $1, 2$-th terms of (\ref{MGI2'}) respectively.
But the $3$-th term of (\ref{MGI1}) is equal to the $4$-th term of (\ref{MGI2'}) and
 the $4$-th term of (\ref{MGI1}) is equal to the $3$-th term of (\ref{MGI2'}).
Thus by adding (\ref{MGI1}) to (\ref{MGI2'}), the $3, 4$-th terms cancel and we have
\begin{equation}\label{deter1}
\Gamma^{\alpha_{1}\alpha_{2}\alpha_{3}\cdots\alpha_{n}}\Gamma^{(\alpha_{1}+e_{i}+e_{j})(\alpha_{2}+e_{s})(\alpha_{3}+e_{t})\cdots\alpha_{n}}
-\Gamma^{(\alpha_{1}+e_{i}+e_{j})\alpha_{2}\alpha_{3}\cdots\alpha_{n}}\Gamma^{\alpha_{1}(\alpha_{2}+e_{s})(\alpha_{3}+e_{t})\cdots\alpha_{n}}=0.
\end{equation}
(\ref{deter1}) implies that the determinant of the matrix $\left(\begin{smallmatrix}
\Gamma^{\alpha_{1}\alpha_{2}\alpha_{3}\cdots\alpha_{n}}&\Gamma^{\alpha_{1}(\alpha_{2}+e_{s})(\alpha_{3}+e_{t})\cdots\alpha_{n}}\\
\Gamma^{(\alpha_{1}+e_{i}+e_{j})\alpha_{2}\alpha_{3}\cdots\alpha_{n}}&\Gamma^{(\alpha_{1}+e_{i}+e_{j})(\alpha_{2}+e_{s})(\alpha_{3}+e_{t})\cdots\alpha_{n}}
\end{smallmatrix}\right)$ is zero.
This finishes the proof.
\end{proof}

By (\ref{deter1}), we have the following corollary.
\begin{coro}\label{linear-dependent-coro}
For a blockwise symmetric matchgate signature
$\Gamma=(\Gamma^{\alpha_1\alpha_2\cdots\alpha_n})$ and
 any $\alpha_1\alpha_2\cdots\alpha_n\in\{0, 1\}^{n\ell}$ and $1\leq i < j\leq\ell, 1\leq s, t\leq \ell$,
\begin{equation*}
\left(\begin{smallmatrix}
\Gamma^{\alpha_{1}\alpha_{2}\alpha_{3}\cdots\alpha_{n}}\\
\Gamma^{(\alpha_{1}+e_{i}+e_{j})\alpha_{2}\alpha_{3}\cdots\alpha_{n}}
\end{smallmatrix}\right)~~
{\rm and}~~
\left(\begin{smallmatrix}
\Gamma^{\alpha_{1}(\alpha_{2}+e_{s})(\alpha_{3}+e_{t})\cdots\alpha_{n}}\\
\Gamma^{(\alpha_{1}+e_{i}+e_{j})(\alpha_{2}+e_{s})(\alpha_{3}+e_{t})\cdots\alpha_{n}}
\end{smallmatrix}\right)
\end{equation*}
are linearly dependent.
\end{coro}

\begin{lemma}\label{xxx-2}
Let $\Gamma=(\Gamma^{\alpha_{1}\alpha_{2}\cdots\alpha_{n}})$ be a blockwise symmetric matchgate signature of arity $n\ell$, where $n\geq 3$.
Then for any $\alpha_{1}\alpha_{2}\cdots\alpha_{n}\in\{0, 1\}^{n\ell}$,
\begin{equation*}
B=\left(\begin{smallmatrix}
\Gamma^{\alpha_{1}\alpha_{2}\alpha_{3}\cdots\alpha_{n}}&\Gamma^{\alpha_{1}(\alpha_{2}+e_{s}+e_{t})\alpha_{3}\cdots\alpha_{n}}\\
\Gamma^{(\alpha_{1}+e_{i}+e_{j})\alpha_{2}\alpha_{3}\cdots\alpha_{n}}&
\Gamma^{(\alpha_{1}+e_{i}+e_{j})(\alpha_{2}+e_{s}+e_{t})\alpha_{3}\cdots\alpha_{n}}
\end{smallmatrix}\right),
\end{equation*}
is degenerate, where $e_i, e_j, e_s, e_t\in\{0, 1\}^{\ell}$ and $i<j$ and $ s<t$.
\end{lemma}
\begin{proof}
We will prove that $\det(B)=0$ by applying MGI.

Assume that $\det(B)\neq 0$, firstly, we claim that
\[\Gamma^{(\alpha_{1}+e_{j})(\alpha_{3}+e_{u})(\alpha_{2}+e_{s}+e_{t})\cdots\alpha_{n}}=
 \Gamma^{(\alpha_{1}+e_{i})\alpha_{2}(\alpha_{3}+e_{v})\cdots\alpha_{n}}=0\] for any $1\leq u, v\leq\ell$.
If there is $u$ such that $\Gamma^{(\alpha_{1}+e_{j})(\alpha_{3}+e_{u})(\alpha_{2}+e_{s}+e_{t})\cdots\alpha_{n}}\neq 0$,
then $\Gamma^{(\alpha_{2}+e_{s}+e_{t})(\alpha_{1}+e_{j})(\alpha_{3}+e_{u})\cdots\alpha_{n}}\neq 0$ since $\Gamma$
us blockwise symmetric. Thus
\begin{equation}\label{nonzero-vector}
\left(\begin{smallmatrix}
\Gamma^{\alpha_{2}(\alpha_{1}+e_{j})(\alpha_{3}+e_{u})\cdots\alpha_{n}}\\
\Gamma^{(\alpha_{2}+e_{s}+e_{t})(\alpha_{1}+e_{j})(\alpha_{3}+e_{u})\cdots\alpha_{n}}
\end{smallmatrix}\right)
\end{equation}
is not a zero vector.
Since
\begin{center}
$\left(\begin{smallmatrix}
\Gamma^{\alpha_{2}\alpha_{1}\alpha_{3}\cdots\alpha_{n}}\\
\Gamma^{(\alpha_{2}+e_{s}+e_{t})\alpha_{1}\alpha_{3}\cdots\alpha_{n}}
\end{smallmatrix}\right)$ and
$\left(\begin{smallmatrix}
\Gamma^{\alpha_{2}(\alpha_{1}+e_{j})(\alpha_{3}+e_{u})\cdots\alpha_{n}}\\
\Gamma^{(\alpha_{2}+e_{s}+e_{t})(\alpha_{1}+e_{j})(\alpha_{3}+e_{u})\cdots\alpha_{n}}
\end{smallmatrix}\right)$,
\end{center}
\begin{center}
$\left(\begin{smallmatrix}
\Gamma^{\alpha_{2}(\alpha_{1}+e_{i}+e_{j})\alpha_{3}\cdots\alpha_{n}}\\
\Gamma^{(\alpha_{2}+e_{s}+e_{t})(\alpha_{1}+e_{i}+e_{j})\alpha_{3}\cdots\alpha_{n}}
\end{smallmatrix}\right)$ and
$\left(\begin{smallmatrix}
\Gamma^{\alpha_{2}(\alpha_{1}+e_{j})(\alpha_{3}+e_{u})\cdots\alpha_{n}}\\
\Gamma^{(\alpha_{2}+e_{s}+e_{t})(\alpha_{1}+e_{j})(\alpha_{3}+e_{u})\cdots\alpha_{n}}
\end{smallmatrix}\right)$
\end{center}
are linearly dependent respectively by Corollary~\ref{linear-dependent-coro}
and the nonzero vector in (\ref{nonzero-vector}) appears as the second vector in both cases. Thus
\begin{center}
$\left(\begin{smallmatrix}
\Gamma^{\alpha_{2}\alpha_{1}\alpha_{3}\cdots\alpha_{n}}\\
\Gamma^{(\alpha_{2}+e_{s}+e_{t})\alpha_{1}\alpha_{3}\cdots\alpha_{n}}
\end{smallmatrix}\right)$ and
$\left(\begin{smallmatrix}
\Gamma^{\alpha_{2}(\alpha_{1}+e_{i}+e_{j})\alpha_{3}\cdots\alpha_{n}}\\
\Gamma^{(\alpha_{2}+e_{s}+e_{t})(\alpha_{1}+e_{i}+e_{j})\alpha_{3}\cdots\alpha_{n}}
\end{smallmatrix}\right)$
\end{center}
are linearly dependent.
This contradicts that
$\det(B)\neq 0$.
Thus $\Gamma^{(\alpha_{1}+e_{j})(\alpha_{3}+e_{u})(\alpha_{2}+e_{s}+e_{t})\cdots\alpha_{n}}=0$ for any $1\leq u\leq\ell$.
By
replacing $\alpha_{2}+e_{s}+e_{t}$ with $\alpha_2$, replacing $\alpha_1+e_j$ with $\alpha_1+e_i$
and replacing $\alpha_3+e_u$ with $\alpha_3+e_v$ respectively,
we can prove that $\Gamma^{(\alpha_{1}+e_{i})(\alpha_{3}+e_{v})\alpha_{2}\cdots\alpha_{n}}=0$ for any $v\in[\ell]$
in the same way.
Thus $\Gamma^{(\alpha_{1}+e_{i})\alpha_{2}(\alpha_{3}+e_{v})\cdots\alpha_{n}}=0$ for any $1\leq v\leq \ell$,
i.e.,
 we have
\[\Gamma^{(\alpha_{1}+e_{j})(\alpha_{3}+e_{u})(\alpha_{2}+e_{s}+e_{t})\cdots\alpha_{n}}=
 \Gamma^{(\alpha_{1}+e_{i})\alpha_{2}(\alpha_{3}+e_{v})\cdots\alpha_{n}}=0\] for any $1\leq u, v\leq\ell$.

Now we apply MGI again.
Let the pattern be $(\alpha_{1}+e_{i})\alpha_{2}\alpha_{3}\alpha_{4}\cdots\alpha_{n}$
and the position vector be $(e_{i}+e_{j})(\alpha_{2}+\alpha_{3})(\alpha_{2}+\alpha_{3}+e_{s}+e_{t})\mathbf{0}\cdots \mathbf{0}$
and $S=\{k| $the $k$-th bit of $\alpha_2+\alpha_3$ is 1$\}$,
$T=\{k| $the $k$-th bit of $\alpha_2+\alpha_3+e_s+e_t$ is 1$\}$, then
by MGI we have
\begin{equation}\label{MGI3}
\begin{split}
&\Gamma^{\alpha_{1}\alpha_{2}\alpha_{3}\cdots\alpha_{n}}
\Gamma^{(\alpha_{1}+e_{i}+e_{j})\alpha_{3}(\alpha_{2}+e_{s}+e_{t})\cdots\alpha_{n}}
-\Gamma^{(\alpha_{1}+e_{i}+e_{j})\alpha_{2}\alpha_{3}\cdots\alpha_{n}}
\Gamma^{\alpha_{1}\alpha_{3}(\alpha_{2}+e_{s}+e_{t})\cdots\alpha_{n}}\\
&+\displaystyle\sum_{u\in S}(\pm \Gamma^{(\alpha_{1}+e_{i})(\alpha_{2}+e_{u})\alpha_{3}\cdots\alpha_{n}}
\Gamma^{(\alpha_{1}+e_{j})(\alpha_{3}+e_{u})(\alpha_{2}+e_{s}+e_{t})\cdots\alpha_{n}})\\
&+\displaystyle\sum_{v\in T}(\pm \Gamma^{(\alpha_{1}+e_{i})\alpha_{2}(\alpha_{3}+e_{v})\cdots\alpha_{n}}
\Gamma^{(\alpha_{1}+e_{j})\alpha_{3}(\alpha_{2}+e_{s}+e_{t}+e_{v})\cdots\alpha_{n}})=0.
\end{split}
\end{equation}
Since $\Gamma^{(\alpha_{1}+e_{j})(\alpha_{3}+e_{u})(\alpha_{2}+e_{s}+e_{t})\cdots\alpha_{n}}=
 \Gamma^{(\alpha_{1}+e_{i})\alpha_{2}(\alpha_{3}+e_{v})\cdots\alpha_{n}}=0$ for all $1\leq u, v\leq \ell$,
we have
\begin{equation}\label{mgi-detB=0}
\Gamma^{\alpha_{1}\alpha_{2}\alpha_{3}\cdots\alpha_{n}}
\Gamma^{(\alpha_{1}+e_{i}+e_{j})\alpha_{3}(\alpha_{2}+e_{s}+e_{t})\cdots\alpha_{n}}
-\Gamma^{(\alpha_{1}+e_{i}+e_{j})\alpha_{2}\alpha_{3}\cdots\alpha_{n}}
\Gamma^{\alpha_{1}\alpha_{3}(\alpha_{2}+e_{s}+e_{t})\cdots\alpha_{n}}=0
\end{equation}
 by (\ref{MGI3}).
Since $\Gamma$ is blockwise symmetric, we have
\begin{equation*}
\Gamma^{\alpha_{1}\alpha_{2}\alpha_{3}\cdots\alpha_{n}}
\Gamma^{(\alpha_{1}+e_{i}+e_{j})\alpha_{3}(\alpha_{2}+e_{s}+e_{t})\cdots\alpha_{n}}
-\Gamma^{(\alpha_{1}+e_{i}+e_{j})\alpha_{2}\alpha_{3}\cdots\alpha_{n}}
\Gamma^{\alpha_{1}\alpha_{3}(\alpha_{2}+e_{s}+e_{t})\cdots\alpha_{n}}=0
\end{equation*}
 by (\ref{mgi-detB=0}).
Thus we have $\det(B)=0$ and finishes
the proof.
\end{proof}

\begin{theorem}\label{mainlemma}
If $\Gamma=(\Gamma^{\alpha_{1}\alpha_{2}\cdots\alpha_{n}})$ is a blockwise symmetric matchgate signature of arity $n\ell$, where $n\geq 3$, then
rank$(M(\Gamma))\leq 2$.
\end{theorem}
\begin{proof}
If rank$(M(\Gamma))\geq 3$, by Lemma \ref{le-col-3-1},
then there is a full rank submatrix of $M(\Gamma)$ of the following form
\begin{equation*}
A=\left(\begin{smallmatrix}
\Gamma^{\alpha_{1}\alpha_{2}\alpha_{3}\cdots\alpha_{n}}&\Gamma^{\alpha_{1}(\alpha_{2}+e_{s})(\alpha_{3}+e_{t})\cdots\alpha_{n}}\\
\Gamma^{(\alpha_{1}+e_{i}+e_{j})\alpha_{2}\alpha_{3}\cdots\alpha_{n}}&\Gamma^{(\alpha_{1}+e_{i}+e_{j})(\alpha_{2}+e_{s})(\alpha_{3}+e_{t})\cdots\alpha_{n}}
\end{smallmatrix}\right)
\end{equation*}
or
\begin{equation*}
B=\left(\begin{smallmatrix}
\Gamma^{\alpha_{1}\alpha_{2}\alpha_{3}\cdots\alpha_{n}}&\Gamma^{\alpha_{1}(\alpha_{2}+e_{s}+e_{t})\alpha_{3}\cdots\alpha_{n}}\\
\Gamma^{(\alpha_{1}+e_{i}+e_{j})\alpha_{2}\alpha_{3}\cdots\alpha_{n}}&
\Gamma^{(\alpha_{1}+e_{i}+e_{j})(\alpha_{2}+e_{s}+e_{t})\alpha_{3}\cdots\alpha_{n}}
\end{smallmatrix}\right),
\end{equation*}
where
$e_i, e_j, e_s, e_t\in\{0, 1\}^{\ell}$ and $i<j$ and $s<t$.
But by Lemma~\ref{xxx-1} and Lemma~\ref{xxx-2}, we know
both  $A$ and $B$ are degenerate.
This is a contradiction.
\end{proof}

Now we can prove Theorem~\ref{equality-no-matchgate}.
\begin{proof}
Let $M(=_n)$ be the matrix form of the {\sc Equality} signature $(=_n)$. Then $M(=_n)_{i, i\cdots i}=1$ for $0\leq i\leq q-1$ and all other entries of
$M(=_n)$ are zero. Thus rank$(M(=_n))=q$. For any $q\times 2^{\ell}$ matrix with rank $q$,
let $\Gamma=(=_n) M^{\otimes n}$, then $\Gamma$ is blockwise symmetric by Lemma~\ref{symmetry} and
$M(\Gamma)=M^{T}(M(=_n) M^{\otimes n-1}$ by Lemma~\ref{signature-matrix}.
Since rank$(M(=_n))=q\geq 3$, rank$(M^T)=q$ and rank$(M^{\otimes n-1})=q^{(n-1)}$. We have rank$(M(\Gamma))\geq 3$ by Lemma~\ref{rank}.
If $\Gamma$ is realized by a matchgate, i.e., $\Gamma$ is a matchgate signature, then
by  Theorem~\ref{mainlemma}, rank$(M(\Gamma))\leq 2$.
This is a contradiction.
\end{proof}
We remark that
the condition $n\geq 3$ is necessary in Theorem~\ref{mainlemma}.
For example, let
$\Gamma=(\Gamma^{\alpha_{1}\alpha_{2}})$, where $\alpha_{1}, \alpha_{2}\in\{0, 1\}^{2}$, with
$\Gamma^{0000}=\Gamma^{1001}=\Gamma^{0110}=1, \Gamma^{1111}=-1$ and all other entries
are zero. Note that $\Gamma$ is blockwise symmetric and is realized by the matchgate
$\Gamma_1$.
But rank$(M(\Gamma))=4$.

\setlength{\unitlength}{2.3mm}
\begin{picture}(15,15)(-30,-6)
\put(0, 0){\line(1,0){8}}
\put(-0.3, -0.3){$\bullet$}
\put(7.7, -0.3){$\bullet$}
\put(0, 8){\line(1,0){8}}
\put(-0.3, 7.7){$\bullet$}
\put(7.7, 7.7){$\bullet$}
\put(0, 0){\line(0, 1){8}}
\put(8, 0){\line(0, 1){8}}
\put(0, 4){\line(1,0){8}}
\put(-0.3, 3.7){$\bullet$}
\put(7.7, 3.7){$\bullet$}
\put(3, 4.5){\rm \scriptsize -1}
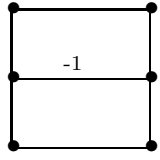
\put(-25, -2){\rm \scriptsize $\Gamma_1$:
The four nodes at the corner are external nodes and the other two nodes are internal nodes.}
\put(-22, -4){\rm \scriptsize The middle edge has weight -1 and all other edges have weight 1.}
\end{picture}

Theorem~\ref{equality-no-matchgate} implies that the {\sc Equalities} cannot be realized by matchgates
under holographic transformation. Thus we have the following theorem.
\begin{theorem}
For any $q\geq 3$, there is no polynomial time algorithms for $\operatorname{Pl-\#CSP}(\mathcal{F})$
$(\operatorname{Pl-Holant}(\mathcal{EQ}|\mathcal{F}))$,
for any $\mathcal{F}$ such that $\operatorname{\#CSP}(\mathcal{F})$ is $\operatorname{\#P}$-hard, that is obtained by a holographic
transformation $M\in\mathbb{C}^{q\times 2^{\ell}}$ of rank $q$,
such that all {\sc Equalities} in $\mathcal{EQ}$ are transformed to $\mathscr{M}$.
\end{theorem}

\section{The structure of blockwise symmetric matchagate signatures}
For an $n$-bit string $\alpha\in\{0, 1\}^{n}$, we define $p(\alpha)=0$ if wt$(\alpha)$ is even and $p(\alpha)=1$ if wt$(\alpha)$ is odd.

\begin{definition}{\rm
For an even (resp. odd) matchgate $\Gamma$ with arity $n$, the {\it condensed signature} $(g_{\alpha})$  of $\Gamma$
is a vector of dimension $2^{n-1}$, and $g_{\alpha}=\Gamma^{\alpha b}$ (resp. $g_{\alpha}=\Gamma^{\alpha\bar{b}}$),
where $\alpha\in\{0, 1\}^{n-1}$ and $b=p(\alpha)$, and we say the condensed signature $g$ realized by the matchgate $\Gamma$}
\end{definition}

We prove the following theorem that characterizes the structure of blockwise symmetric matchgate signatures.
\begin{theorem}\label{fu'stheorem1}
Let $\Gamma=(\Gamma^{\alpha_{1}\alpha_2\cdots\alpha_n}$) be a blockwise symmetric sigature
with arity $n\ell$, where $n\geq 3$.
If $\Gamma$ is realized by a matchgate,
then  there exists a condensed signature $(g_{\alpha})_{\alpha\in\{0, 1\}^{\ell}}$ that is realized by a matchgate with arity $\ell+1$,
and a bitwise symmetric matchgate signature $\Gamma_{S}$ that is realized by a matchgate with arity $n$ such that
\begin{equation}\label{xxxx}
\Gamma^{\alpha_{1}\alpha_{2}\cdots\alpha_{n}}
=g_{\alpha_{1}}g_{\alpha_{2}}\cdots g_{\alpha_{n}}\Gamma_{S}^{p(\alpha_{1})p(\alpha_{2})\cdots p(\alpha_{n})}.
\end{equation}
\end{theorem}

In \cite{cailu10}, the following similar theorem was given.
\begin{theorem}\label{cai'stheorem1}
Let $\Gamma=(\Gamma^{\alpha_{1}\alpha_2\cdots\alpha_n})$ be a blockwise symmetric signature
with arity $n\ell$, where $n\geq 4$.
If $\Gamma$ is realized by an even matchgate (resp. odd matchgate) and $\Gamma^{\mathbf{0}\mathbf{0}\cdots \mathbf{0}}\neq 0$
(resp. $\Gamma^{e_1 \mathbf{0}\cdots \mathbf{0}}\neq 0$),
then  there exists a condensed signature $(g_{\alpha})_{\alpha\in\{0, 1\}^{\ell}}$ that is realized by a matchgate with arity $\ell+1$, and a bitwise symmetric matchgate signature $\Gamma_{S}$ that is realized by a matchgate with arity $n$ such that
\begin{equation*}
\Gamma^{\alpha_{1}\alpha_{2}\cdots\alpha_{n}}
=g_{\alpha_{1}}g_{\alpha_{2}}\cdots g_{\alpha_{n}}\Gamma_{S}^{p(\alpha_{1})p(\alpha_{2})\cdots p(\alpha_{n})}.
\end{equation*}
\end{theorem}

Note that we improve Theorem~\ref{cai'stheorem1} by removing the non-vanishing condition
$\Gamma^{\mathbf{0}\mathbf{0}\cdots \mathbf{0}}\neq 0$
(or $\Gamma^{e_1 \mathbf{0}\cdots \mathbf{0}}\neq 0$), and improving
$n\geq 4$ to $n\geq 3$.
Moreover, the counterexample given by $\Gamma_1$ in the end of Section~\ref{main-section} shows that the condition $n\geq 3$ is necessary.

For a blockwise symmetric matchgate signature $\Gamma=(\Gamma^{\alpha_{1}\alpha_2\cdots\alpha_n}$)
with arity $n\ell$, we have
 rank$(M(\Gamma))=0, 1$ or 2 by Lemma~\ref{mainlemma}.
If rank$(M(\Gamma))=0$, the proof is trivial. In the following, we prove Theorem \ref{fu'stheorem1} for rank($M(\Gamma)$)=2 and omit the proof for
rank($M(\Gamma)$)=1 due to the space limit.

\begin{proof}
By Lemma~\ref{le-col-2-1}, $M(\Gamma)$ has a full rank submatrix of the following form
$\left(\begin{smallmatrix}
\Gamma^{\theta\gamma_{2}\cdots\gamma_{n}}&0\\
0&\Gamma^{\eta(\gamma_{2}+e_{t})\cdots\gamma_{n}}
\end{smallmatrix}\right)$
or $\left(\begin{smallmatrix}
0&\Gamma^{\theta\gamma_{2}\cdots\gamma_{n}}\\
\Gamma^{\eta(\gamma_{2}+e_{t})\cdots\gamma_{n}}&0
\end{smallmatrix}\right),$
where $\theta, \eta, \gamma_j\in\{0, 1\}^{\ell}$ for $2\leq j\leq n$ and $\theta+\eta=e_{s}$,
where $e_s, e_t\in\{0, 1\}^{\ell}$.
Note that $p(\theta)\neq p(\eta)$.
Without loss of generality, we assume that $p(\theta)=0$ and $p(\eta)=1$ in the following.

Note that the rows $M(\Gamma)^{\theta}$ and $M(\Gamma)^{\eta}$ are linearly independent.
Then by rank$(M(\Gamma))=2$, all of the rows of $M(\Gamma)$ are linear combinations of  $M(\Gamma)^{\theta}$ and $M(\Gamma)^{\eta}$.
Moreover,  for any $\alpha_i, \alpha_j\in\{0, 1\}^{\ell}$, by the parity condition of $\Gamma$, if $p(\alpha_i)\neq p(\alpha_j)$,
the nonzero entries of $M(\Gamma)^{\alpha_{i}}$ and $M(\Gamma)^{\alpha_{j}}$
are in disjoint column positions.
So $M(\Gamma)^{\alpha_{i}}$ and $M(\Gamma)^{\alpha_{j}}$ are orthogonal. Thus for any $\alpha\in\{0, 1\}^{\ell}$
there exists $g_{\alpha}$  such that
$M(\Gamma)^{\alpha}=g_{\alpha} M(\Gamma)^{\theta}$ if $p(\alpha)= p(\theta)$,
and $M(\Gamma)^{\alpha}=g_{\alpha} M(\Gamma)^{\eta}$ if $p(\alpha)= p(\eta)$.

Up to a global nonzero scalar, we can assume that $\Gamma^{\theta\gamma_2\cdots\gamma_n}=1$ and $\Gamma^{\eta(\gamma_2+e_t)\cdots\gamma_n}=r$, then
for any $\alpha\in\{0, 1\}^{\ell}$,
$g_{\alpha}=\Gamma^{\alpha\gamma_2\cdots\gamma_n}$ if $p(\alpha)=p(\theta)$, and
$g_{\alpha}=r^{-1}\Gamma^{\alpha(\gamma_2+e_t)\cdots\gamma_n}$ if $p(\alpha)=p(\eta)$.
Now we prove that $g=(g_{\alpha})$ is a condensed signature.
Assume that $\Gamma=(\Gamma^{\alpha_1\alpha_2\cdots\alpha_n})$ is realized by the matchgate
$G$ with arity $n\ell$.
Note that $G$ has $n\ell$ external nodes
and we group the external nodes into $n$ blocks of
consecutive $\ell$ nodes each.
Firstly, we connect a path $L$ of length 2 to
the $(\ell+t)$-th external node, i.e., the $t$-th node of the second block, of $G$.
We denote the $(\ell+t)$-th external node of $G$, i.e., one endpoint of $L$, as $v_1$,
another endpoint of $L$ as $v_3$ and the middle node of $L$ as $v_2$.
If the $(\ell+t)$-th bit of $\theta\gamma_2\cdots\gamma_n$ is 0, we give the edge between
$v_1$ and $v_2$ weight $r^{-1}$ and the edge between
$v_2$ and $v_3$ weight 1.
If the $(\ell+t)$-th bit of $\theta\gamma_2\cdots\gamma_n$ is 1, we give the edge between
$v_1$ and $v_2$ weight 1 and the edge between
$v_2$ and $v_3$ weight $r^{-1}$, and view $v_3$ as an external node and $v_1, v_2$ as internal nodes.
We do nothing to the external nodes in the  first block and still view them as external nodes.
We view other external nodes, i.e., the external nodes that are not in the first block and not the $(\ell+t)$-th external node,
as internal node and do the following operations to them:
for $\ell+1\leq i\leq n\ell$ and $i\neq \ell+t$,
 if the $i$-th bit of $\theta\gamma_2\cdots\gamma_n$
is 0, do nothing to $G$; if the $i$-th bit of $\theta\gamma_2\cdots\gamma_n$
is 1, connect an edge with weight 1 to the $i$-th
external node of $G$ and view the new nodes as internal nodes.
Then we get a new matchgate $G'$ with arity $\ell+1$ and $g$ is the
condensed signature of $G'$.

Now we construct the matchgate signature $\Gamma_S$ of arity $n$.
For every block of $G$, for  $1\leq i\leq \ell$,
 if the $i$-th bit of $\theta$ is 1
then  we add an edge
 of weight 1 to the $i$-th external node,
and the new node replaces it as an external node.
 If the $i$-th bit of $\theta$ is 0
then we do nothing to it.
We get a new matchgate $H$.
Note that all the bits of $\theta, \eta$ are the same except the $s$-th bit.
Next, we define  $H'$ from  $H$:
 In each block of  $\ell$ external nodes of $H$,
we pick only the $s$-th external node as an external node
of $H'$; all others are considered internal nodes of
$H'$.
We denote the matchgate signature of  $H'$
by $\Gamma_{S}$.
Then
 $\Gamma_S=(\Gamma_{S}^{j_1 j_2 \cdots j_n})$
is a signature of arity $n$
and
$\Gamma_S^{j_1j_2\cdots j_n}=\Gamma^{\alpha_1 \alpha_2\cdots \alpha_n}$, where $\alpha_i\in\{\theta, \eta\}$ and
$j_i=p(\alpha_i)$.

Now for any $\alpha_1\alpha_2\cdots\alpha_n\in\{0, 1\}^{n\ell}$, we prove (\ref{xxxx}) by  induction on $t$,
 the number of $\alpha_i$ from $\{\alpha_1, \alpha_2, \cdots, \alpha_n\}$ that do not belong to $\{\theta, \eta\}$ in $\{\alpha_1, \alpha_2, \ldots, \alpha_n\}$.
Note that $g_{\theta}=g_{\eta}=1$.
If all $\alpha_i$ belong to $\{\theta, \eta\}$, i.e., $t=0$, we are done by the definition of $\Gamma_{S}^{j_1 j_2 \cdots j_n}$
and (\ref{xxxx}) is proved.
Inductively we assume that (\ref{xxxx}) has been proved for $t-1$
and there are $t\geq 1$ blocks in $\{\alpha_1, \alpha_2, \ldots, \alpha_n\}$
that that do not belong to  $\{\theta, \eta\}$.
Then
we pick one
  $\alpha_i$  that do not belong to $\{\theta, \eta\}$.
If  $\alpha_i$ has the same parity as $\theta$, then
the $\alpha_i$-th row $M(\Gamma)^{\alpha_i}=g_{\alpha_i}M(\Gamma)^{\theta}$
in
$M(\Gamma)=\left(\begin{smallmatrix}
\cdots&\cdots&\cdots\\
\cdots&\Gamma^{\theta\alpha_1\cdots\alpha_{i-1}\alpha_{i+1}\cdots\alpha_n}&\cdots\\
\cdots&\cdots&\cdots\\
\cdots&\Gamma^{\alpha_i\alpha_1\cdots\alpha_{i-1}\alpha_{i+1}\cdots\alpha_n}&\cdots\\
\cdots&\cdots&\cdots
\end{smallmatrix}\right).$
Thus
 $\Gamma^{\alpha_1\alpha_2\cdots\alpha_n}=\Gamma^{\alpha_i\alpha_1\cdots\alpha_{i-1}\alpha_{i+1}\cdots\alpha_n}=
g_{\alpha_i}\Gamma^{\theta\alpha_1\cdots\alpha_{i-1}\alpha_{i+1}\cdots\alpha_n}.$
Note that
the number of blocks from $\{\theta, \alpha_1, \cdots, \alpha_{i-1}, \alpha_{i+1}, \cdots, \alpha_n\}$ that do not belong to $\{\theta, \eta\}$ is
$t-1$. Thus we have
$
\Gamma^{\theta\alpha_1\cdots\alpha_{i-1}\alpha_{i+1}\cdots\alpha_n}=g_{\alpha_2}\cdots g_{\alpha_{i-1}}g_{\alpha_{i+1}}\cdots g_{\alpha_n}
\Gamma_S^{p(\theta) p(\alpha_2)\cdots p(\alpha_n)}
$
by induction.
Then
$\Gamma^{\alpha_1\alpha_2\cdots\alpha_n}=
g_{\alpha_1}g_{\alpha_2}\cdots g_{\alpha_n}
\Gamma_S^{p(\alpha_1)p(\alpha_2)\cdots p(\alpha_n)}$
 and we complete the proof.
If $p(\alpha_i)=p(\eta)$, the proof is similar and we omit it here.
\end{proof}


\renewcommand{\refname}{References}

\end{document}